\documentclass[11pt]{llncs}

\usepackage[latin5]{inputenc}

\usepackage{amssymb}
\usepackage{amsmath}
\usepackage{amsthm}
\usepackage{textcomp}
\usepackage{array}
\usepackage{multirow}

\usepackage{color}
\usepackage[left=4cm,top=3cm,right=4cm,bottom=3cm]{geometry}

\newcommand{\ket}[1]{|#1\rangle}

\newcommand{\Cent}[0]{\mbox{\textcent}}
\newcommand{\dollar}[0]{\$}
\newtheorem{openproblem}{Open Problem}
\newtheorem{fact}{Fact}

\title{Quantum counter automata\thanks{A preliminary version of this paper appeared as \textit{A.~C.~Cem Say, Abuzer Yakary{\i}lmaz, and {\c S}efika Y\"{u}zsever.
\newblock Quantum one-way one-counter automata.
\newblock In R\={u}si\c{n}\v{s} Freivalds, editor, {\em Randomized and quantum
  computation}, pages 25--34, 2010 (\newblock Satellite workshop of MFCS and CSL 2010).}}}
\author{A.C. Cem Say$ ^{\mbox{\tiny 1,}} $\thanks{Say's work was partially supported by the Scientific and Technological
Research Council of Turkey (T\"{U}B\.ITAK) with grant 108E142.} \and Abuzer Yakary{\i}lmaz$ ^{\mbox{\tiny 2,}} $\thanks{Yakary{\i}lmaz was partially supported by the Scientific and Technological
Research Council of Turkey (T\"{U}B\.ITAK) with grant 108E142 and FP7 FET-Open project QCS.}}
\institute{$ ^{\mbox{\tiny 1}} $Bo\u{g}azi\c{c}i University, Department of Computer Engineering, Bebek 34342 \.{I}stanbul, Turkey
\\ \email{say@boun.edu.tr} ~~ \\
$ ^{\mbox{\tiny 2}} $University of Latvia, Faculty of Computing, Raina bulv. 19, Riga, LV-1586, Latvia \\
\email{abuzer@boun.edu.tr}
 \\~~\\
\today
}

\begin{document}

\newlength{\twidth}
\maketitle
\pagenumbering{arabic}

%-----------------------------------------------------------------------------%
\begin{abstract} \label{abstract:Abstract}
%-----------------------------------------------------------------------------%

The question of whether quantum real-time one-counter automata (rtQ1CAs) can outperform their probabilistic counterparts has been open for more than a decade. We provide an affirmative answer to this question, by demonstrating a non-context-free language that can be recognized with perfect soundness by a rtQ1CA. This is the first demonstration of the superiority of a quantum model to the corresponding classical one in the real-time case with an error bound less than 1.
We also introduce a generalization of the rtQ1CA, the quantum one-way one-counter automaton (1Q1CA), and show that they too are superior to the corresponding family of probabilistic machines. For this purpose, we provide  general definitions of these models that reflect the modern approach to the definition of quantum finite automata, and point out some problems with previous results. We identify several remaining open problems.

\end{abstract}

%-----------------------------------------------------------------------------%
\section{Introduction} \label{section:Introduction}
%-----------------------------------------------------------------------------%
Although a complete understanding of the relationship between the polynomial time complexity classes corresponding to classical and quantum computers seems to be a distant goal, a restricted version of this question for constant-memory machines has already been answered: Linear-time quantum finite automata (QFAs) that are allowed to pause for some steps on a symbol in a single left-to-right scan of the input string can solve some problems for which  probabilistic machines, even with two-way access to their input,  require exponential time \cite{AI99,DS90}. Interestingly, when these automata are further restricted to perform real-time access to the input, i.e. forbidden to pause, the probabilistic and quantum models have identical language recognition power \cite{Bo03,Je07,YS11A}. To our knowledge, no quantum automaton model has yet been shown to outperform its probabilistic counterpart in terms of language recognition with one-sided error bound less than 1 in the real-time mode, which corresponds to the smallest possible nontrivial time bound. In this paper, we give the first demonstration of such a superiority in the case of the real-time one-counter automaton model.

One-counter automata can simply be thought of as finite automata enhanced by the addition of a single integer counter of unlimited capacity. Instructions in the programming language of these machines can increment or decrement this counter, and test its value for being zero, in addition to the standard state transition actions of finite automata.
The study of quantum  real-time one-counter automata 
was initiated by Kravtsev \cite{Kr99}, who based his definition on a popular QFA model of the time, introduced by Kondacs and Watrous \cite{KW97}. Kravtsev provided some example machines recognizing certain languages, all of which were later shown by Yamasaki \textit{et al.} \cite{YKTI02} to be also recognizable by classical probabilistic reversible real-time automata. 

It is now accepted \cite{Hi10,YS11A} that the Kondacs-Watrous QFA model has been defined in an unnecessarily restricted way, and does not utilize the full flexibility provided by quantum physics. As a result, these QFAs are strictly less powerful than even classical deterministic finite automata \cite{KW97}, and this weakness also affected the Kravtsev model of quantum counter automata, with Yamasaki \textit{et al.} even demonstrating \cite{YKTI02} a regular language which they could not recognize with bounded error. On the other hand, Bonner \textit{et al.} \cite{BFK01} claimed to demonstrate a Kravtsev machine recognizing a particular language that could not be recognized by any probabilistic real-time one-counter automaton.

We should also note that Yamasaki \textit{et al.} \cite{YKI05} studied quantum one-counter automata with two-way access to their input, and  have shown that the class of languages they  recognize with bounded error contains some languages that are unrecognizable by  deterministic two-way one-counter automata that are restricted to perform a fixed number of counter reversals. The fairer question about whether these two-way quantum machines can outperform their \textit{probabilistic} counterparts, which are also allowed to make bounded error, remains open.

The modern approach to the definition of quantum computer models \cite{Wa03,YS11A} has as an easy consequence that any quantum machine can simulate its probabilistic counterpart efficiently, and the real question is whether the quantum version can outperform the classical one or not. In this paper, 
we first provide a general definition of the quantum real-time one-counter
automaton (rtQ1CA) that reflects this modern approach. In an earlier version of this manuscript \cite{SYY10}, we pointed out that the above-mentioned result of Bonner \textit{et al.} about the relationship of the quantum and classical real-time one-counter automaton models is flawed, and identified the related question as still being open. Here, we provide our own proof about a different language where the quantum model is indeed superior to its probabilistic counterpart. We then define a new model, the quantum one-way one-counter automaton, and prove a stronger result about the comparative powers of quantum and classical machines of this type. It turns out that the ability to pause the read head on the tape for some steps can be used to perform error reduction in quantum one-way machines. We also make some observations about the relationship of two-way quantum counter automata with  some seemingly more restricted models.

%-----------------------------------------------------------------------------%
\section{Real-Time One-Counter Automata} \label{section:PFA-QA}
%-----------------------------------------------------------------------------%
We start with clarifying a potentially confusing matter of terminology. A real-time machine is one that moves its input tape head to the right at each step of its execution, terminating after executing its last instruction on the right input end-marker \cite{Ra63B}. A one-way machine also performs a single left-to right scan of the input, but has the additional capacity of pausing its head for some steps on an input symbol \cite{SHL65}. The source of the confusion is that some authors, who do not consider the latter case, have used the designation 
``one-way" for what we call real-time machines, and what we call one-way machines have variously been called ``on-line" or ``1.5-way."  This section is on real-time counter automata. We will define and study general one-way automata, for the first time, in Section \ref{section:onewaysuperiority}.

A real-time one-counter automaton (rt1CA) can be seen as a real-time finite state
machine augmented with a counter which can be modified by an amount
from the set $ \lozenge = \{ -1, 0 , +1 \} $, and where the sign of this counter is also taken
into account during transitions. 
A configuration of a rt1CA is represented by $ (q,k) $, where $ q \in Q
$, and $ k \in \mathbb{Z} $
represents the value of the counter.

A probabilistic rt1CA (probabilistic real-time one-counter automaton --
rtP1CA) is a 5-tuple
\[
       \mathcal{P} = (Q,\Sigma,\delta,q_{1},Q_{a}),
\]
where $ Q $ denotes the set of internal states, $\Sigma$ is the input alphabet, not containing the end-marker symbols
$ \Cent $ and $ \dollar $, $ q_{1} \in Q $ is the initial state, 
$ Q_{a} \subseteq Q $ is the set of accept states, and
$ \delta $ is known as the transition function. We  define $ \tilde{\Sigma} = \Sigma \cup \{ \Cent, \dollar \} $.

The probability that a rtP1CA which reads symbol $ \sigma
\in \tilde{\Sigma} $ when it is in state $ q \in Q $, and when its
counter
has sign $ \theta \in \Theta = \{ 0,\pm \}$, will switch to state $ q^{\prime} \in Q $ and
modify the value of the counter by $ c \in \lozenge $ is
\[
       \delta(q,\sigma,\theta,q^{\prime},c) \in \mathbb{R},
\]

The transition function $ \delta $ must satisfy the following condition of well-formedness, dictated by the elementary rules of probability:
For any choice of $ q \in Q $, $ \sigma \in \tilde{ \Sigma } $ , and $ \theta \in \Theta $,
\[
       \sum\limits_{q^{\prime} \in Q, c \in \lozenge }
\delta(q,\sigma,\theta,q^{\prime},c) = 1.
\]

Given an input string $ w \in \Sigma $, a rtP1CA scans the tape content $ \tilde{w} = \Cent w \dollar  $ from left to right, and accepts the input if it is in an accept state after
reading the right end-marker $ \dollar $. The probability that machine $ \mathcal{P} $ will accept the input $ w $ is denoted by	$ f_{\mathcal{P}}^{a}(w) $. The probability that $ \mathcal{P} $ will reject $ w $ is denoted by	$ f_{\mathcal{P}}^{r}(w) $.

It is  easier to construct well-formed rt1CAs which obey the additional restriction that the counter increment $ c \in \lozenge $ associated with each transition to state $
q^{\prime}$ upon reading tape symbol $
\sigma $ is determined by the pair $
(q^{\prime},\sigma)$, for any $
(q^{\prime},\sigma) \in Q \times \tilde{\Sigma} $. Such machines are said to be ``simple" \cite{Kr99,BFK01,YKTI02,YKI05}.
We denote the above-mentioned relation by the function $ D_{c} : Q \times \tilde{\Sigma}
\rightarrow \lozenge $.
The 
counter increment argument is dropped from the definition of the transition function of simple machines,\footnote{In fact, choosing the domain of $ D_{c} $ as $ Q $ instead of
$ Q \times \tilde{\Sigma} $ does not decrease the computational power of the
``simple" models that we will define shortly. However, this is not known to apply for the KQ1CA, the Kravtsev model of rtQ1CA,
which is the variant originally defined in the literature.} so the well-formedness condition for simple rtP1CAs is as follows:
For any choice of $ q \in Q $, $ \sigma \in \tilde{ \Sigma } $ , and $ \theta \in \Theta $,
\[
       \sum\limits_{q^{\prime} \in Q}
\delta(q,\sigma,\theta,q^{\prime}) = 1.
\]
Transition functions of simple rtP1CAs can be represented by using a left stochastic transition matrix for
each $ \sigma \in \tilde{\Sigma} $
and $ \theta \in \Theta $, say, $ A_{\sigma,\theta} $, whose rows and columns
are indexed by internal
states, and $ A_{\sigma,\theta}[j,i] $ (i.e. the $ (j,i)^{th} $ entry of $ A_{\sigma,\theta} $) represents the probability of
transition  from state $ q_{i} $
to $ q_{j} $ upon reading symbol $ \sigma $ with counter sign $ \theta
$ (thereby incrementing 
the counter by $ D_{c}(q_{j},\sigma) $). 

\begin{fact}\label{fact:simpp}
       For any rtP1CA $ \mathcal{P}_{1} $,  there exists a simple rtP1CA $ \mathcal{P}_{2} $,
       such that for all $ w \in \Sigma^{*} $,
       \[
               f_{\mathcal{P}_{1}}^{a}(w) = f_{\mathcal{P}_{2}}^{a}(w).
       \]
\end{fact}

We now present our modern definition of the quantum versions of these machines. 

A quantum rt1CA
(quantum real-time one-counter automaton -- rtQ1CA) is a 7-tuple
\[
       \mathcal{M} = (Q,\Sigma,\Omega,\delta,q_{1},\Omega_{a}, \omega_{1}),
\]
where the new items in the definition reflect a slight change in the architecture of the machine: rtQ1CAs have an additional finite register, capable of holding a symbol from the alphabet $ \Omega $. $ \Omega_{a} \subseteq \Omega$. $ \omega_{1} \in \Omega $
		is   the initial register symbol. This register, which is updated by each transition in the program, is automatically reset to the value $ \omega_{1}$ before each step. This simple mechanism is sufficient
to implement superoperators \cite{YS11A} that are allowed by quantum theory, and provides the machines with the required flexibility.

The amplitude with which a rtQ1CA that reads symbol $ \sigma
\in \tilde{\Sigma} $ when it is in state $ q \in Q $, and when its
counter
has sign $ \theta \in \Theta $, will switch to state $ q^{\prime} \in Q $,
modify the value of the counter by $ c \in \lozenge $, and write $ \omega \in \Omega $ to the finite register is
\[
       \label{equation:delta-quantum}
               \delta(q,\sigma,\theta,q^{\prime},c,\omega) \in \mathbb{C}.
\]

A rtQ1CA can be in a superposition of
more than one configuration at the same time. At any step, the computation may split probabilistically to as many branches as there are distinct register symbols, with each branch represented by such a superposition. The ``weight" $\alpha_{\kappa}$
of each configuration $\kappa$ in these superpositions is called its amplitude, and equals the sum of the products of the amplitudes of the configurations that performed transitions to $\kappa$ in the previous step, and the amplitudes of the associated transitions. The superposition is expressed by the sum $\sum_{\kappa \in (Q\times \mathbb{Z})} \alpha_{\kappa}\ket{\kappa}$.

In order to be consistent with the restrictions imposed by quantum theory, the transition function $ \delta $ must 
 satisfy the following local conditions of well-formedness:\footnote{These turn out to be straightforward generalizations of the corresponding conditions in \cite{Kr99}.}
For any choice of $ q_{1},q_{2} \in Q $, $ \sigma \in \tilde{\Sigma} $, and
$ \theta_{1},\theta_{2} \in \Theta $,

\begin{equation}
	\label{equation:wfc1}
       \sum\limits_{q^{\prime} \in Q, c \in \lozenge, \omega \in \Omega } \mspace{-27mu}
               \overline{
               \delta(q_{1},\sigma,\theta_{1},q^{\prime},c,\omega ) }
               \delta(q_{2},\sigma,\theta_{1},q^{\prime},c,\omega )
               = \left\lbrace
               \begin{array}{lll}
                       1 ~~ &  \mbox{if }q_{1} = q_{2},  \\
                       0 &  \mbox{otherwise,}
               \end{array}
               \right.
\end{equation}

\begin{equation}\label{equation:wfc2}
       \begin{array}{ll}
               \displaystyle \sum\limits_{q^{\prime} \in Q, \omega \in \Omega }
               &
               \overline{\delta(q_{1},\sigma,\theta_{1},q^{\prime},+1,\omega)}
               \delta(q_{2},\sigma,\theta_{2},q^{\prime}, 0 ,\omega)
               \\
               & + ~
               \overline{\delta(q_{1},\sigma,\theta_{1},q^{\prime},0,\omega)}
               \delta(q_{2},\sigma,\theta_{2},q^{\prime}, -1 ,\omega) = 0,
       \end{array}     
\end{equation}
and
\begin{equation}\label{equation:wfc3}
               \displaystyle \sum\limits_{q^{\prime} \in Q, \omega \in \Omega }
               \overline{\delta(q_{1},\sigma,\theta_{1},q^{\prime},+1,\omega)}
               \delta(q_{2},\sigma,\theta_{2},q^{\prime}, -1 ,\omega) = 0.
\end{equation}

Upon processing the right end-marker, the finite register is observed using a projective measurement \cite{NC00}, which yields the outcome ``a" if the symbol written in the register is in $ \Omega_{a}$, and the outcome ``r" otherwise.
The probability that machine $ \mathcal{M} $ will accept the input $ w $, that is,	$ f_{\mathcal{M}}^{a}(w) $, is defined as the probability that ``a" will be observed as a result of this process. For each probabilistic branch, this is just the sum of the squares of the moduli of the amplitudes of the configurations which were reached by transitions writing symbols in $\Omega_{a}$ in the register at the end of the computation. See \cite{YS11A} for a detailed explanation of this procedure for related automaton models.

For simple rtQ1CAs, one only needs to consider a single
local condition for well-formedness:
For any choice of $ q_{1},q_{2} \in Q $, $ \sigma \in \tilde{\Sigma} $, and
$ \theta \in \Theta $,
\[
	\mspace{-1mu}
       \sum\limits_{q^{\prime} \in Q,  \omega \in \Omega } \mspace{-15mu}
               \overline{
               \delta(q_{1},\sigma,\theta,q^{\prime},\omega ) }
               \delta(q_{2},\sigma,\theta,q^{\prime},\omega )
               = \left\lbrace
               \begin{array}{lll}
                       1 ~~ &  \mbox{if }q_{1} = q_{2},  \\
                       0 &  \mbox{otherwise.}
               \end{array}
               \right. 
\]

Analogously to simple rt1PCAs, transition functions of simple rtQ1CAs can also be represented in matrix notation:
For each $
\sigma \in \tilde{\Sigma} $
and $ \theta \in \Theta $, one defines a superoperator $
\mathcal{E}_{\sigma,\theta}$, described by a collection $\{
E_{\sigma,\theta,\omega}\} $
for $\omega \in \Omega  $, where  $ E_{\sigma,\theta,\omega}[j,i] $ represents
the amplitude of the transition  from state $ q_{i} $
to $ q_{j} $ upon reading symbol $ \sigma $, sensing counter sign $ \theta $,
and writing $ \omega $ on the finite register
(thereby updating the value of the counter by $ D_{c}(q_{j},\sigma) $).
The resulting machine is well-formed if $ \mathcal{E}_{\sigma,\theta} $ is admissible \cite{Wa03}, i.e.
\[
       \sum_{\omega \in \Omega} E_{\sigma,\theta,\omega}^{\dagger}
E_{\sigma,\theta,\omega} = I.
\]
It is an open problem whether the computational powers of general rtQ1CAs and
simple rtQ1CAs are the same or not.

In this context, machines obeying Kravtsev's original definition \cite{Kr99}, which we call KQ1CAs, are rtQ1CA variants where $\Omega $ is partitioned to three pairwise disjoint subsets $ \Omega_{n} $, $
\Omega_{a} $, and $ \Omega_{r} $, $| \Omega_{n} | = |
\Omega_{a} | = | \Omega_{r} | = 1 $, the initial register symbol is the element of $ \Omega_{n} $,
the finite register is observed after each transition, with outcomes ``n", ``a", and ``r", corresponding to these subsets, and the machine halts whenever ``a" or ``r" is observed, or the tape string has been completely traversed. It is easy to show that the KQ1CA's additional capability of halting without traversing the complete input gives it no superiority over the general rtQ1CA, since KQ1CA algorithms that use this feature have equivalent standard rtQ1CA algorithms which employ a somewhat larger register alphabet.

A quantum automaton with blind counter is one that
never checks the status of its counter,  and accepts its input only if the counter is zero,
and the processing of the input has ended by writing an accept symbol in the register. The $\Theta$ argument is therefore
removed from $\delta$ in the definitions of these machines.

We conclude this section by introducing some more basic notational items and terminology that will be used in the rest of the paper.

The language $ L \subseteq \Sigma^{*} $ recognized by machine $
\mathcal{M} $ with cutpoint
$ \lambda \in \mathbb{R} $ is defined as
\[
      L = \{ w \in \Sigma^{*} \mid f_{\mathcal{M}}^{a}(w) > \lambda \}.
\]

Machines that recognize their languages with  cutpoint 0 are said to be nondeterministic.\footnote{This property, ensuring that strings not in the language are guaranteed to be rejected with probability 1, is also called perfect soundness.} 

A language $L$ is said to be recognized with one-sided error bound $\epsilon$ if there exists a nondeterministic machine that recognizes $L$, accepting all members of $L$ with probability at least $1 - \epsilon $, where $\epsilon$ is a real constant such that $ 0 \le \epsilon < 1 $.
The language $ L \subseteq \Sigma^{*} $ is said to be recognized by
machine $ \mathcal{M} $ with (two-sided) error bound $ \epsilon $
($ 0 \le \epsilon < \frac{1}{2} $) if
\begin{itemize}
      \item $ f_{\mathcal{M}}^{a}(w) \ge 1 - \epsilon $ when $ w \in L $,
      \item $ f_{\mathcal{M}}^{r}(w) \ge 1 - \epsilon $ when $ w \notin L $.
\end{itemize}
This situation is also known as recognition with bounded error.

For a given string $ w $, $ |w| $ denotes the length of $ w $, 
$ |w|_{\sigma} $ is the number of occurrences of symbol $ \sigma $ in $ w $,
and $ w_{i} $ denotes the $ i^{th} $ symbol of $ w $.
	
%-----------------------------------------------------------------------------%
\section{Computational Power of Quantum One-Counter Automata} \label{section:main-results}
%-----------------------------------------------------------------------------%

We start by showing that our sufficiently general definition allows
rtQ1CAs to simulate rtP1CAs easily.

\begin{theorem}
      \label{theorem:rtP1CA-rtQ1CA}
      For any rtP1CA $ \mathcal{P} $, there exists a simple rtQ1CA $
\mathcal{M} $,
      such that for all $ w \in \Sigma^{*} $,
      \[
              f_{\mathcal{P}}^{a}(w) = f_{\mathcal{M}}^{a}(w).
      \]
\end{theorem}
\begin{proof}
      Assume that $ \mathcal{P} $ is a simple rtP1CA with transition matrices $A_{\sigma,\theta}$ by Fact \ref{fact:simpp}. The state set $Q$ and initial state of $ \mathcal{M} $ will be identical to those of $
\mathcal{P} $. $\Omega$ contains a symbol $
\omega_{i,j} $ for each $ 1 \le i,j \le |Q| $. Any one of these can be designated to be the initial register symbol.
      The admissible operators specifying the transitions of $ \mathcal{M} $ are constructed as follows: For each $ \sigma \in
\tilde{\Sigma} $, $ \theta \in \Theta $, and $ \omega_{i,j} \in \Omega $,
      \[
              E_{\sigma,\theta,\omega_{i,j}}[j,i] = \sqrt{A_{\sigma,\theta}[j,i]}, 
      \]
     and all remaining entries of $ E_{\sigma,\theta,\omega_{i,j}} $ are zero. 
     
Since the entries in the columns of the $A_{\sigma,\theta}$ sum up to 1, it is easy to see that all the diagonal entries of $
\sum_{\omega \in \Omega} E_{\sigma,\theta,\omega}^{\dagger}
E_{\sigma,\theta,\omega}$ are 1's, and therefore     $ \mathcal{M} $ is well-formed.

     $\Omega_{a}$ is specified as the set of all symbols $
\omega_{i,j} $, where $q_{j}$ is an accept state of $ \mathcal{P} $. 

      $ \mathcal{M} $ traces precisely the same computational paths with the same probabilities as $ \mathcal{P} $, and accepts whenever $ \mathcal{P} $ accepts.
\end{proof}

We now turn to the question of finding a case where rtQ1CAs outperform rtP1CAs.\footnote{Nakanishi \textit{et al.} \cite{NHK06} demonstrate such a superiority of quantum over classical in the case of pushdown automata, which are natural generalizations of rt1CAs, however, the machine they construct is not real-time. Furthermore, it uses multiple stack symbols, and is therefore not a counter automaton.} Let us review one past attempt to solve this problem.

Let $ \Sigma=\{a,b,0,1,2,3,4,5\} $, $ \Sigma_{a}=\{a,0,1,2\} $, $ \Sigma_{b} = \{b,3,4,5\} $.
We define a homomorphism  $ h_{\Sigma^{\prime}} $ from $ \Sigma^{*} $
to $ (\Sigma^{\prime})^{*} $ as
\begin{itemize}
      \item $ h_{\Sigma^{\prime}}(\sigma) = \sigma $ if $ \sigma \in \Sigma^{\prime} $, and
      \item $ h_{\Sigma^{\prime}}(\sigma) = \varepsilon $ if $ \sigma \in \Sigma \setminus
\Sigma^{\prime} $, where $ \varepsilon $ is the empty string.
\end{itemize}

Consider the languages $ L_{1} $ and $ L_{2} $, defined in Figures \ref{figure:L-1} and \ref{figure:L-2}. In \cite{BFK01}, Bonner \textit{et al.} claim that 
 $ L_{2} $
 cannot be recognized with bounded error by any
rtP1CA. They go on to demonstrate a KQ1CA that recognizes $ L_{2} $ with
error bound $ \frac{1}{4} $. Together with our simulation result
(Theorem \ref{theorem:rtP1CA-rtQ1CA}),
this would establish the superiority of
rtQ1CAs over rtP1CAs, however, the argument in \cite{BFK01} about $ L_{2} $'s
unrecognizability by rtP1CAs is unfortunately flawed. A brief
description of this problem follows:

\begin{figure}[here]      
      \fbox{
      \begin{minipage}{0.92\textwidth}
              \[
                      L_{1} = ( L_{a,1} \cap L_{b,2} ) \cup ( L_{a,2}
\cap L_{b,1} ),
              \]
              where
              \scriptsize
              \begin{eqnarray*}
                      L_{a,1} & = & \{ w \in \Sigma^{*} \mid
h_{\Sigma_{a}}(w)=xay, x \in
\{0,1\}^{*}, y \in \{2\}^{*},
                      | h_{\{0\}}(x) | = | h_{\{1\}}(x) | \}
                      \\
                      L_{a,2} & = & \{ w \in \Sigma^{*} \mid
h_{\Sigma_{a}}(w)=xay, x \in
\{0,1\}^{*}, y \in \{2\}^{+},
                      | h_{\{0\}}(x) | = | h_{\{1\}}(x) | + |y| \}
                      \\
                      L_{a,3} & = & \{ \Sigma^{*} \setminus ( L_{a,1}
\cup L_{a,2} ) \}
                      \\
                      L_{b,1} & = & \{ w \in \Sigma^{*} \mid
h_{\Sigma_{b}}(w)=xby, x \in
\{3,4\}^{*}, y \in \{5\}^{*},
                      | h_{\{3\}}(x) | = | h_{\{4\}}(x) | \}
                      \\
                      L_{b,2} & = & \{ w \in \Sigma^{*} \mid
h_{\Sigma_{b}}(w)=xby, x \in
\{3,4\}^{*}, y \in \{5\}^{+},
                      | h_{\{3\}}(x) | = | h_{\{4\}}(x) | + |y| \}
                      \\
                      L_{b,3} & = & \{ \Sigma^{*} \setminus ( L_{b,1}
\cup L_{b,2} ) \}
              \end{eqnarray*}
              \normalsize
      \end{minipage}
      }
      \caption{The description of $ L_{1} $}
      \label{figure:L-1}
\end{figure}

\begin{figure}[here]
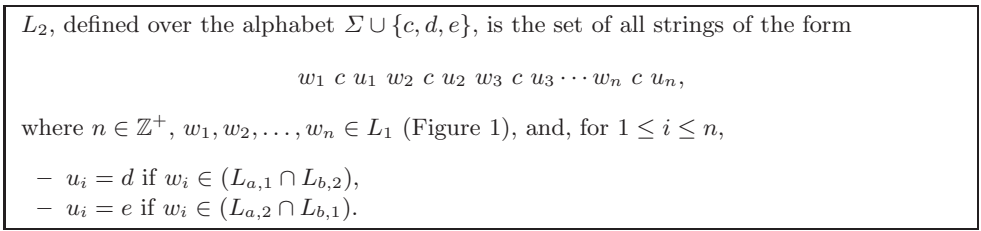
      
      \fbox{
      \begin{minipage}{0.92\textwidth}
              $ L_{2} $, defined over the alphabet $ \Sigma \cup
\{c,d,e\} $, is the set of all strings of the form
              \[
                      w_{1} ~ c ~ u_{1} ~ w_{2} ~ c ~ u_{2} ~ w_{3} ~
c ~ u_{3} \cdots
w_{n} ~ c ~ u_{n},
              \]
              where $ n \in \mathbb{Z}^{+} $, $ w_{1}, w_{2}, \ldots, w_{n} \in
L_{1} $ (Figure 1), and, for $ 1 \le i \le n $,
              \begin{itemize}
                      \item $ u_{i} = d $ if $ w_{i}  \in ( L_{a,1}
\cap L_{b,2} ),  $
                      \item $ u_{i} = e $ if $ w_{i}  \in ( L_{a,2}
\cap L_{b,1} ).  $
              \end{itemize}
      \end{minipage}
      }
      \caption{The description of $ L_{2} $}
      \label{figure:L-2}
\end{figure}

Bonner \textit{et al.} start by proving that no deterministic rt1CA can
recognize the language $ L_{1} $, and note  that $ L_{1} $ can therefore not be
recognized with zero error by a rtP1CA either.
They then state that
\\\\
\indent
\begin{minipage}{0.9\textwidth}
\textit{``The impossibility of recognizing $ L_{2} $ by a
probabilistic automaton with a bounded
error now follows, since the subwords $ w_{i} \in L_{1} $  of a word $
w \in L_{2} $ can be taken in
arbitrarily large numbers, and every $ w_{i} $ is processed with a
positive error."}
\end{minipage}
\\\\
The problem with this argument is that it does not consider the
possibility that there might exist a rtP1CA that recognizes $ L_{1} $ with
negative one-sided error, e.g. a machine $ \mathcal{M} $ that responds to input
strings $ w $ according to the following schema:
If $ w \in L_{1}  $, $ \mathcal{M} $ accepts $ w $ with probability 1.
If $ w\notin L_{1}  $, $ \mathcal{M} $ accepts $ w $ with probability at most
$ \frac{1}{4} $, and rejects $ w $ with
probability at least $ \frac{3}{4} $.

Such a machine would go past the argument in \cite{BFK01}, so presently there is no working proof of the unrecognizability of  $L_{2} $  by rtP1CAs. Note that the discovery of a negative answer to
\begin{openproblem}
      Is the complement of $L_{1}$ context-free?
\end{openproblem}
would fix this issue.
 
Interestingly, it is easy to demonstrate a superiority of quantum over probabilistic for the following restricted version of the
counter automaton model, which may strike the reader as being somewhat silly at the first sight.
A real-time  increment-only counter automaton (rtIOCA) is a rt1CA whose program neither checks nor decrements the value of the counter. As expected, the probabilistic versions of these machines (rtPIOCAs) are equivalent in power to probabilistic finite automata. Using quantum effects that will also be utilized in our algorithms in this paper, Yakary{\i}lmaz \textit{et al.} \cite{YFSA11A} have proven the following facts about quantum real-time increment-only counter automata (rtQIOCAs):
\begin{fact}	
	The class of languages recognized with bounded error by rtQIOCAs contains all languages recognized with 
	bounded error by standard rtQ1CAs with blind counter. These include $ L_{eq} = \{ w \in \{a,b\}^{*} \mid | w |_{a} = | w |_{b} \} $.
\end{fact}	
\begin{fact}
	Any language recognized by a nondeterministic real-time automaton with one blind counter is 
	recognized by a  rtQIOCA with cutpoint $ \frac{1}{2} $.
\end{fact}

\begin{fact}\label{fact:nh}
	Any language recognized by a deterministic 1-reversal rt1CA\footnote{A 1-reversal rt1CA is a machine that is not allowed to increment its counter again after performing the first decrement.} is recognized by a  rtQIOCA with cutpoint $ \frac{1}{2} $. For instance, 	
	\small
      \[ L_{NH} = \{a^{x}ba^{y_{1}}ba^{y_{2}}b \cdots a^{y_{t}}b \mid
x,t,y_{1}, \cdots, y_{t}
      \in \mathbb{Z}^{+} \mbox{ and } \exists k ~ (1 \le k \le t),
x=\sum_{i=1}^{k}y_{i} \},\]
      \normalsize
      over the alphabet $ \{a,b \} $, is such a language.
\end{fact}

Since $ L_{eq} $ is nonregular, no rtPIOCA can recognize it with bounded error. 
Similarly, $ L_{NH} $ is known \cite{NH71} to be 
unrecognizable with cutpoint by (even two-way \cite{Ka91}) 
probabilistic finite automata,\footnote{While reporting some results of his investigation \cite{Ra92} of the 
relationship between the
classes of languages recognized by reversal-bounded deterministic
one-way multicounter machines and bounded-error two-way
probabilistic finite automata (2-PFA), Ravikumar reported, among
others, the following problem as being open:

``Are all languages accepted by 1-way deterministic 1-reversal counter
machines in 2-PFA?"

By virtue of Fact \ref{fact:nh}, this question is answered negatively.} and is therefore unrecognizable by rtPIOCAs and real-time quantum finite automata \cite{YS11A} as well.

%-----------------------------------------------------------------------------%
\subsection{Superiority of rtQ1CAs over rtP1CAs} \label{section:rtsuperiority}
%-----------------------------------------------------------------------------%. 
It is now time to present our main result that unrestricted rtQ1CAs 
are indeed more powerful than their probabilistic counterparts. 

Consider the language
\[
	L_{3} = \{ w \in \{a,b,c\}^{*} \mid |w|_{a} \neq |w|_{b} \mbox{ and }  
	( |w|_{a} = |w|_{c} \mbox{ or } |w|_{b} = |w|_{c} ) \}.
\]

$ L_{3}$ is not context-free, which can be shown easily using the pumping lemma on the string $a^{p!+p}b^{p}c^{p!+p}$. This means that no probabilistic pushdown automaton (and therefore no rtP1CA) can recognize $ L_{3}$ with one-sided error. 

\begin{theorem}\label{theorem:abuth-unordered}
    There exists a KQ1CA with blind counter that recognizes $ L_{3} $ with one-sided
	error bound $ \frac{3}{4} $.
\end{theorem}
\begin{proof}

We construct a KQ1CA with blind counter 
$ \mathcal{M}_{1}=(Q,\Sigma,\Omega,\delta,q_{1},\Omega_{a},\omega_{n}) $,
where $ Q = \{q_{1},p_{1}\}$, $ q_{1} $ is the initial state, $ \Sigma=\{a,b,c\}, $ 
$ \Omega=\{\omega_{n},\omega_{a},\omega_{r} \}$, and $\Omega_{a}=\{\omega_{a}\} $.
The transition function $\delta$ is depicted in Figure \ref{fig:L3-unordered}. 

\begin{figure}[h]	
	\centering
	\mbox{
	\small
	\begin{minipage}{\textwidth}		
		\[
			\begin{array}{|l|l|l|l|}
				\hline
				\multicolumn{1}{|c|}{} &
				\multicolumn{3}{c|}{ \Cent}
				\\
				\hline
				\multicolumn{1}{|c|}{} &
				\multicolumn{3}{l|}{	
					\delta(q_{1},\Cent) =\frac{1}{\sqrt{2}} (q_{1},0,\omega_{n})
					+ \frac{1}{\sqrt{2}} (p_{1},0,\omega_{n})
				}
				\\
				\hline \hline
				\multicolumn{1}{|c|}{} &
				\multicolumn{1}{c|}{a} &
				\multicolumn{1}{c|}{b} &
				\multicolumn{1}{c|}{c} 
				\\ \hline
				\mathsf{path_{1}} 
				&
					\begin{array}{rcl}
						\delta(q_{1},a) & = & (q_{1},+1,\omega_{n}) \\						
					\end{array}
				 &
					\begin{array}{rcl}
						\delta(q_{1},b) & = &  (q_{1},0,\omega_{n}) \\
					\end{array}
				&
					\begin{array}{rcl}
						\delta(q_{1},c) & = & (q_{1},-1,\omega_{n}) \\					 
					\end{array}
				\\
				\hline
				\mathsf{path_{2}} 
				&
					\begin{array}{rcl}
						\delta(p_{1},a) & = & (p_{1},0,\omega_{n}) \\						
					\end{array}
				 &
					\begin{array}{rcl}
						\delta(p_{1},b) & = &  (p_{1},+1,\omega_{n}) \\
					\end{array}
				&
					\begin{array}{rcl}
						\delta(p_{1},c) & = & (p_{1},-1,\omega_{n}) \\					 
					\end{array}
				\\				
				\hline \hline
				\multicolumn{1}{|c|}{} &
				\multicolumn{3}{c|}{ \dollar}
				\\
				\hline
				\multicolumn{1}{|c|}{\mathsf{path_{1}}} &
				\multicolumn{3}{l|}{
					\begin{array}{lcl}
						\delta(q_{1},\dollar) & = & 
							\frac{1}{\sqrt{2}} (q_{1},0,\omega_{n}) + \frac{1}{\sqrt{2}} (p_{1},0,\omega_{a})
					\end{array}
				}
				\\
				\hline
				\multicolumn{1}{|c|}{\mathsf{path_{2}}} &
				\multicolumn{3}{l|}{
					\begin{array}{lcl}
						\delta(p_{1},\dollar) & = & 
							\frac{1}{\sqrt{2}} (q_{1},0,\omega_{n}) - \frac{1}{\sqrt{2}} (p_{1},0,\omega_{a})
					\end{array}	
					
				}
				\\
				\hline
			\end{array}			
		\]
	\end{minipage}
	}
	\caption{The program of KQ1CA $\mathcal{M}_{1}$ which recognizes $L_{3}$}
	\label{fig:L3-unordered}
\end{figure}

In the figure, transitions applicable in the case where $ \mathcal{M}_{1} $ is in state $ q \in Q $ and reading
	symbol $ \sigma \in \tilde{\Sigma} $ are represented using the notation
	\[
		\delta(q,\sigma) =  \sum_{(q',c,\omega) \in Q  \times \lozenge \times \Omega }
			\delta(q,\sigma,0,q',c,\omega) (q',c,\omega).
	\]
	It is easy to see that the amplitudes of the unspecified transitions for $\delta(p_{1},\Cent)$ can be filled in so that conditions (\ref{equation:wfc1}-\ref{equation:wfc3})  are satisfied. Note that, for any combination of $q, q' \in Q$, $\sigma \in \Sigma$, $c \in \lozenge$, and $\omega \in \Omega$, $\delta(q,\sigma,0,q',c,\omega)=\delta(q,\sigma,\pm,q',c,\omega)$, since $ \mathcal{M}_{1} $ is a blind counter machine.

$\mathcal{M}_{1}$ starts by splitting into two computational paths, $ \mathsf{path}_{1} $
and $ \mathsf{path}_{2} $, with equal
                       amplitude. $ \mathsf{path}_{1} $ computes the difference $|w|_{a} - |w|_{c}$ by incrementing the counter for each $a$ that it scans, ignoring the $b$'s, and decrementing for each $c$. $ \mathsf{path}_{2} $ similarly computes  $|w|_{b} - |w|_{c}$. Upon reaching the right end-marker $\dollar$, the machine is in superposition
\[
	\frac{1}{\sqrt{2}} \ket{q_{1},(|w|_{a} - |w|_{c})} + \frac{1}{\sqrt{2}} \ket{p_{1}, (|w|_{b} - |w|_{c}) },
\]
and the two paths perform a two-way quantum Fourier transform (QFT) \cite{KW97}, described using the transitions for $\delta(q_{1},\dollar)$ and $\delta(p_{1},\dollar)$ in Figure \ref{fig:L3-unordered}, yielding the final superposition
\footnotesize
\begin{equation}
	\label{eq:L3-superposi}
	\frac{1}{2}  \ket{q_{1},(|w|_{a} - |w|_{c})} + \frac{1}{2} \ket{q_{1},(|w|_{b} - |w|_{c})} +  
	\frac{1}{2} \ket{p_{1},(|w|_{a} - |w|_{c})} - \frac{1}{2} \ket{p_{1},(|w|_{b} - |w|_{c})}.
\end{equation}
\normalsize
      If $ |w|_{a} = |w|_{b} $, the configuration with state $p_{1}$ is ``cancelled out."  Since the only configuration in the final superposition is reached by writing a register symbol not in $\Omega_{a}$, the input string $w$ is rejected with probability 1.
      If $ |w|_{a} \neq |w|_{b}  $, of the four configurations in 
(\ref{eq:L3-superposi}), the first two have written the rejecting symbol, as discussed above. The remaining two configurations have written the accepting symbol. If $|w|_{a} = |w|_{c} \mbox{ or } |w|_{b} = |w|_{c}$, exactly one of these computational paths will have its counter equal zero, and therefore accept. Otherwise, all paths will reject. An examination of the amplitudes reveals that members of $L_{3}$ are accepted with probability $ \frac{1}{4} $, and all nonmembers are rejected with probability 1.

\end{proof}

It is known that real-time quantum Turing machines with cutpoint 0 are strictly more powerful than probabilistic ones for all common sublogarithmic space bounds, but the proof of this statement in \cite{YS10A} holds only for unbounded-error recognition. By Theorems \ref{theorem:rtP1CA-rtQ1CA} and \ref{theorem:abuth-unordered}, the rtQ1CA is the first real-time quantum model that has been shown to outperform its probabilistic counterpart in the regime of recognition with one-sided error bound less than 1.
%-----------------------------------------------------------------------------%
\subsection{One-Way Q1CAs} \label{section:onewaysuperiority}
%-----------------------------------------------------------------------------%

Although the machine of Section \ref{section:rtsuperiority} can  be run repeatedly for a fixed number of times to achieve high correctness probability, its error bound of $\frac{3}{4}$ is still somewhat unsatisfactory, since the core model allows only one left-to-right scan of the input. We do not know of a general method that can be used to convert a given rtQ1CA to another rtQ1CA that recognizes the same language, but with a smaller error bound. In this section, we define the one-way quantum one-counter automaton, (1Q1CA) whose input head can remain stationary for some steps of the computation. We will show that 1Q1CAs outperform their probabilistic  counterparts, the 1P1CAs, with an example where the 1Q1CA can be tuned to recognize its language with any desired positive error probability.

A quantum one-way one-counter automaton (1Q1CA) is an 8-tuple
\[
	\mathcal{M} = (Q,\Sigma,\Omega,\delta,q_{1},\omega_{1},\Omega_{a},\Omega_{r}),
\]
where the new item $ \Omega_{r} \subseteq ( \Omega \setminus \Omega_{a} ) $ is the set of rejecting register symbols. We also define $\Omega_{n} = \Omega \setminus (\Omega_{a} \cup \Omega_{r})$. As will be explained shortly, the definition of the 1Q1CA transition function accommodates the possibility of the tape head pausing on the same symbol for some steps. The halting behavior of machines of this type is also somewhat different than that of rtQ1CAs. 

$\delta$ now has an additional argument that indicates the direction of the head movement associated with the transition in question. The amplitude with which a 1Q1CA that reads symbol $ \sigma \in \tilde{\Sigma} $ 
when it is in state $ q \in Q $, and when its counter
has sign $ \theta \in \Theta $, will switch to state $ q^{\prime} \in Q $,
modify the value of the counter by $ c \in \lozenge $, move the input head in direction $ d \in D$,
and write $ \omega \in \Omega $ in the finite register is
\[
       \label{equation:delta-quantum-one-way}
               \delta(q,\sigma,\theta,q^{\prime},c,d,\omega) \in \mathbb{C},
\]
where $D = \{\downarrow,\rightarrow\}$, with the arrows indicating stationary and rightward moves, respectively.

For a given input $ w \in \Sigma^{*} $,
an instantaneous configuration of 1Q1CA $ \mathcal{M} $ can be specified by the triple $ (q,x,k) $,
where $ q \in Q $, 
$ x $ is the position of the input head ($ 1 \le x \le |\tilde{w}| $),
and $ k $ is the value of the counter.
A list of local conditions for well-formedness of 1Q1CAs is given in  \ref{app:well-formedness}.

After each step of the execution of a 1Q1CA, the register is measured. This measurement has three possible outcomes, ``a", ``r", or ``n", corresponding to whether the register contains a member of $\Omega_{a}$, $\Omega_{r}$, or $\Omega_{n}$, respectively. $\omega_{1} \in \Omega_{n}$. The machine halts with the corresponding response whenever an ``a" or an ``r" is observed. Otherwise, the computation continues.

As with all quantum models defined using the modern approach \cite{YS11A}, 1Q1CAs can simulate their probabilistic counterparts easily. We proceed to our proof of the superiority of the quantum version.

Consider the language
\[
	L_{4} = \{ w \in \{a,b,c\}^{*} \mid |w|_{a} = |w|_{b} \mbox{ and } |w|_{a} \neq |w|_{c} \},
\]
which can be shown to be non-context-free, and therefore unrecognizable by any probabilistic one-way one-counter machine with one-sided error, by an argument almost identical to what we had about $L_{3}$ in Section \ref{section:rtsuperiority}.

\begin{theorem}\label{theorem:cemth}
	For any $N>1$, there exists a 1Q1CA with blind counter that recognizes $ L_{4} $ with one-sided
	error bound $ \frac{1}{N} $.
\end{theorem}
\begin{proof}

We construct a 1Q1CA with blind counter $
\mathcal{M}_{2}=(Q,\Sigma,\Omega,\delta,q_{1},\omega_{n},\Omega_{a},\Omega_{r}) $,
where $ Q=\{q_{1}\} \cup \{ q_{j,k} \mid 1 \leq j \leq N \mbox{ and } 1 \leq k \leq \max\{j+1,N-j+2\} \} 
\cup \{ p_{k} \mid 1 \leq k \leq N \} $, $ \Sigma=\{a,b,c\}$, 
$ \Omega = \{ \omega_{n}, \omega_{a}, \omega_{r} \} $, $\Omega_{a}=\{\omega_{a}\} $, and $\Omega_{r}=\{\omega_{r}\} $.
The transition function $\delta$ is depicted in Figure \ref{fig:L4}, with the amplitudes of all unspecified transitions 
understood to be filled in so that the conditions in  \ref{app:well-formedness} are satisfied.\footnote{It is easy to complete the missing portions of $\delta$, since this particular 1Q1CA has the property that all nonzero-amplitude transitions to any state $q'$ move the head in the same way, satisfying conditions (\ref{equation:appeq4}-\ref{equation:lastapp})  in the list in \ref{app:well-formedness} trivially.}

\begin{figure}[h]	
	\centering
	\mbox{
	\begin{minipage}{\textwidth}		
		\[
			\begin{array}{|l|l|}
				\hline
				\multicolumn{1}{|c|}{} &
				\multicolumn{1}{c|}{\Cent}
				\\ \hline
				& 	\begin{array}{lcl}
						\delta ( q_{1},\Cent ) & = &  \sum\limits_{j=1}^{N} \frac{1}{\sqrt{N}} 
							( q_{j,1},0,\rightarrow,\omega_{n} )
					\end{array}	
				\\
				\hline \hline
				\multicolumn{1}{|c|}{} &
				\multicolumn{1}{c|}{a}
				\\ \hline
				\mathsf{path_{j}~ (1 \le j \le N)} 
								&
					\begin{array}{lcl}
						\delta ( q_{j,1},a ) & = & (q_{j,1},+1,\rightarrow,\omega_{n} ) \\		
					\end{array}
				\\
				\hline \hline
				\multicolumn{1}{|c|}{} &
				\multicolumn{1}{c|}{b} 
				\\ \hline
				\mathsf{path_{j}~ (1 \le j \le N)} 
								&	
					\begin{array}{lcl}
						\delta ( q_{j,1},b ) & = & (q_{j,2},-1,\downarrow,\omega_{n} ) \\	
						\delta ( q_{j,k},b ) & = &  (q_{j,k+1},0,\downarrow,\omega_{n}) ~ (1 < k < j+1) \\
						\delta ( q_{j,j+1},b ) & = &  (q_{j,1},0,\rightarrow,\omega_{n}) \\
					\end{array}
				\\
				\hline \hline
				\multicolumn{1}{|c|}{} &
				\multicolumn{1}{c|}{c} 
				\\
				\hline
				\mathsf{path_{j}~ (1 \le j \le N)} 
								& 
					\begin{array}{lcl}
						\delta ( q_{j,k},c ) & = & ( q_{j,k+1},0,\downarrow,\omega_{n}) ~ ( 0 < k < N-j+2 ) \\
						\delta ( q_{j,N-j+2},c ) & = & ( q_{j,1},0,\rightarrow,\omega_{n})
					\end{array}
				\\ 
				\hline \hline
				\multicolumn{1}{|c|}{} &
				\multicolumn{1}{c|}{\dollar} 
				\\
				\hline
				\mathsf{path_{j}~ (1 \le j \le N)}  
				& 
					\begin{array}{lcl}
						\delta ( q_{j,1},\dollar ) & = & 
							\frac{1}{\sqrt{N}} \sum\limits_{l=1}^{N} e^{\frac{2 \pi i }{N}jl} 
								(p_{l},0,\downarrow,\omega_{n})
						\end{array}
				\\ \hline
					& \begin{array}{lcl}
						\delta( p_{l},\dollar ) & = & (p_{l},0,\downarrow,\omega_{a}) ~ (1 \le l < N) \\
						\delta( p_{N},\dollar ) & = & (p_{N},0,\downarrow,\omega_{r})
					\end{array}
				\\ \hline
			\end{array}			
		\]
	\end{minipage}
	}
	\caption{The program of 1Q1CA $ \mathcal{M}_{2} $ which recognizes $L_{4}$}
	\label{fig:L4}
\end{figure}

Similarly to Figure \ref{fig:L3-unordered}, the transitions applicable in the case where $ \mathcal{M}_{2} $ is in state $ q \in Q $ and reading
	symbol $ \sigma \in \tilde{\Sigma} $ are represented as
	\[
		\delta(q,\sigma) =  \sum_{(q',c,d,\omega) \in Q  \times \lozenge \times D \times \Omega }
			\delta(q,\sigma,0,q',c,d,\omega) (q',c,d,\omega)
	\]
in Figure \ref{fig:L4}.
	
$ \mathcal{M}_{2} $ starts by splitting to $N$ computational paths that traverse the input $w$ with equal amplitude. Each of these paths computes the difference $|w|_{a} - |w|_{b}$ using the counter, as described for the machine $ \mathcal{M}_{1} $ of the previous subsection. The input head passes over the $a$'s without any delay, whereas, for each $j$, ($1 \le j \le N$) $\mathsf{path_{j}}$  remains stationary for $j$ steps over each $b$,   and for $N-j+1$ steps over each $c$, before going on to the next symbol. Each path pauses over the right end-marker for one step to perform an $N$-way version of the QFT, described by the entry for $\delta ( q_{j,1},\dollar )$ in Figure \ref{fig:L4}, and accepts or rejects in the subsequent step, depending on the state entered as a result of this transform.

Let us describe the technique used for comparing the numbers of $b$'s and $c$'s, which was first employed in the celebrated two-way QFA algorithm in \cite{KW97}. Due to the delays on these symbols, $\mathsf{path_{j}}$ reaches the final $\dollar$ at the $(|w|+j|w|_{b}+(N-j+1)|w|_{c}+1)^{th}$ step of the execution. It is easy to see that all paths arrive at this point simultaneously if $|w|_{b}=|w|_{c}$, and no two distinct paths have the same runtime if $|w|_{b} \neq |w|_{c}$. 

If $w \in L_{4}$, the $N$-way QFT is therefore performed separately by each $\mathsf{path_{j}}$:
\[ \frac{1}{\sqrt{N}} \ket{q_{j,1},|\tilde{w}|,0} \rightarrow \sum\limits_{l=1}^{N} \frac{1}{N} 
							 \ket{p_{l},|\tilde{w}|,0} \]
In the final step, $N-1$ of the branches created by the QFT accept, and only one rejects. The acceptance probability contributed by  $\mathsf{path_{j}}$ is $(N-1)(\frac{1}{N})^{2}$. Since there are $N$ paths, $w$ is accepted with probability $\frac{N-1}{N}$.

If $w \notin L_{4}$, there are two possible cases: If $|w|_{a} \neq |w|_{b}$, the counter can not equal zero at the end of any computational branch, and $w$ is rejected with probability 1. If $|w|_{a} = |w|_{b}$ and $|w|_{b} = |w|_{c}$, $ \mathcal{M}_{2} $ is in superposition
\[ \sum\limits_{j=1}^{N} \frac{1}{\sqrt{N}} 
							 \ket{q_{j,1},|\tilde{w}|,0}  \]
just before the QFT, and is transformed to the superposition $\ket{p_{N},|\tilde{w}|,0}$, which also leads to rejection with probability 1.
\end{proof}

%-----------------------------------------------------------------------------%
\subsection{Remarks on Two-Way One-Counter Automata} \label{section:two-way}
%-----------------------------------------------------------------------------%

In the only work on quantum counter automata with two-way access to their input \cite{YKI05}, 
Yamasaki \textit{et al.} demonstrated that the non-context-free languages 
$L_{square} = \{a^{m}b^{m^{2}} \mid m\geq 1\}$,
$L_{prod} = \{a^{m_{1}}b^{m_{2}}c^{m_{1}m_{2}} \mid m_{1},m_{2}\geq 1\}$,
$ L_{power} = \{ a^{m}b^{2^{m}}\mid m\geq 1 \} $, and
$ L_{m,m^{2},\ldots,m^{k}} = \{ a_{1}^{m}a_{2}^{m^{2}}\cdots a_{k}^{m^{k}} \mid m\geq 1 \} $
can be recognized with bounded error 
by quantum two-way one-counter
 automata (2Q1CAs) in polynomial time. In this section, we will point out that these languages can also be recognized by machines whose definitions correspond to two different specializations of the 2Q1CA:
 
It has recently been proven \cite{YS10B} that all  four languages mentioned above, together with a rich collection of other languages that includes $ L_{pal} = \{w \mid w = w ^{reverse} \} $, 
$ L_{twin} =\{ wcw \mid w \in \{a,b\}^{*} \} $, and all \textit{polynomial languages} of the form $ \{a_{1}^{n_{1}} \cdots a_{k}^{n_{k}} b_{1}^{p_{1}(n_{1},\ldots,n_{k})} \cdots 
	b_{r}^{p_{1}(n_{1},\ldots,n_{k})} \mid p_{i}(n_{1},\ldots,n_{k}) \ge 0 \}$, 
where $ a_{1}, \ldots, a_{k},b_{1}, \ldots, b_{r} $ are distinct symbols, and 
each $ p_{i} $ is a polynomial with integer coefficients, can be recognized by two-way quantum finite
automata (2QFAs)  with bounded error in exponential time. Note that a 2QFA is a 2Q1CA without a counter.

In 1994, Petersen showed \cite{Pe94} that $L_{square}$,
$L_{prod}$, and
$ L_{power}$ can be recognized by polynomial-time deterministic two-way one-counter automata (2D1CAs). It is not difficult to see that the machines for $L_{square}$ and
$L_{prod}$  can be combined \cite{YKI05} to obtain a polynomial-time 2D1CA for recognizing 
$ L_{m,m^{2},\ldots,m^{k}}$ as well. The only advantage of the machines of \cite{YKI05} over those of \cite{Pe94} is that the 2Q1CA for $L_{power}$ has linear runtime, whereas the 2D1CA runs in time $O(n \log^{2} n)$.

We therefore note that the following problems are open:
\begin{openproblem}
   Is there a language that can be recognized by a 2Q1CA, but not by a probabilistic two-way one-counter automaton (2P1CA) with bounded error?
\end{openproblem}
\begin{openproblem}
 Is there a language that can be recognized by a bounded-error 2P1CA, but not by a 2D1CA?\end{openproblem}
\begin{openproblem}
   Is there a language that can be recognized by a 2Q1CA, but not by a 2QFA with bounded error?
\end{openproblem}
As mentioned in Fact \ref{fact:nh}, the language $L_{NH}$ can be recognized by a deterministic real-time one-counter automaton. We conjecture that no 2QFA can recognize $L_{NH}$ with bounded error, and that the answer to the last question above is negative.

%-----------------------------------------------------------------------------%
\section{Conclusion} \label{section:ConcludingRemarks}
%-----------------------------------------------------------------------------%
Since counter automata are specializations of pushdown automata (PDAs), Theorem \ref{theorem:abuth-unordered} also establishes the superiority of real-time quantum PDAs over probabilistic ones, extending the result of Nakanishi \textit{et al.} \cite{NHK06}. 
We conclude with the following questions.
\begin{openproblem}
   Is there a non-context-free language that can be recognized with one-sided error bound less than 
   $ \frac{3}{4} $ by a rtQ1CA?
\end{openproblem}
\begin{openproblem}
Is there a language that can be recognized by a rtQ1CA, but not by a rtP1CA with (two-sided) bounded error? In particular, can a rtQ1CA recognize the language 
$ L_{IP}= \{uv^{reverse} \in \{0, 1\}^{*} \mid |u| = |v|, \sum_{i=1}^{|u|} u_i \cdot v_i = 1 \mod 2 \} $ \cite{HS10}?
\end{openproblem}
\begin{openproblem}
   Is there a language that can be recognized by a 1Q1CA, but not by a one-way QFA with bounded error?
\end{openproblem}

\appendix

%-----------------------------------------------------------------------------%
\section{Local Conditions for Well-Formedness of a 1Q1CA} \label{app:well-formedness}
%-----------------------------------------------------------------------------%

Let $ \mathcal{M} = (Q,\Sigma,\Omega,\delta,q_{1},\Omega_{a},\Omega_{r},\omega_{1}) $ be a 1Q1CA.
Suppose that the configurations of $\mathcal{M}$ are indexed by positive integers.
Let $ v_{i}^{\omega} $ be an infinite-dimensional vector such that
$ v_{i}^{\omega}[k] $ is the amplitude of the transition
from the $ i^{th} $ configuration to the $ k^{th} $ configuration
that writes $ \omega \in \Omega $ in the finite register
if the $ k^{th} $ configuration is reachable from the $ i^{th} $ configuration
in a single step, and is zero otherwise.
A well-formed machine must satisfy\footnote{The reader is referred to \cite{YS11A} for a detailed treatment of this issue.} 
\begin{equation}
	\label{eq:orthonormal-vectors}
	\sum_{\omega \in \Omega} (v_{i}^{\omega})^{\dagger} \cdot v_{j}^{\omega} =
	\left\lbrace \begin{array}{ll}
		1 & \mbox{ if } i=j, \\
		0 & \mbox{ otherwise}
	\end{array} \right.
\end{equation}
for all $i$, $j$.
Note that $ (v_{i}^{\omega}[k])^{*} \cdot v_{j}^{\omega}[k] =0 $ 
unless both the $ i^{th} $ and the $ j^{th} $ configurations yield  the $ k^{th} $ configuration in a single step.

Let $(q_{1},x_{1},k_{1}) $ and $ (q_{2},x_{2},k_{2}) $ be two configurations.
These can possibly evolve to the same configuration in one step only if $ | x_{1}-x_{2} | \leq 1 $ or $ | k_{1}-k_{2} | \leq 2 $.
By considering all possible combinations of these differences, we obtain the following list of restrictions on $ \delta $ in order for it to satisfy
Eq. (\ref{eq:orthonormal-vectors}).
For any choice of $ q_{1},q_{2} \in Q $; $ \sigma,\sigma_{1},\sigma_{2} \in \tilde{\Sigma} $; 
$ \theta, \theta_{1}, \theta_{2} \in \Theta $; and $ d_{1},d_{2} \in D $:
\\
1. $ x_{1} = x_{2} $ and $ k_{1} = k_{2} $:
\begin{equation}
	\sum\limits_{q' \in Q, c \in \lozenge,d \in D, \omega \in \Omega  }
		\overline{ \delta(q_{1},\sigma,\theta,q',c,d,\omega) } 
		\delta(q_{2},\sigma,\theta,q',c,d,\omega)
		= \left\lbrace \begin{array}{ll}
                       1 ~~ &  \mbox{if }q_{1} = q_{2},  \\
                       0 &  \mbox{otherwise.}
		\end{array}  \right. 
\end{equation}
2. $ x_{1} = x_{2} $ and $ k_{1} = k_{2}-1 $ ($ x_{1} = x_{2} $ and $ k_{1} = k_{2}+1 $):
\begin{equation}
	\begin{array}{ll}
		\sum\limits_{q' \in Q, d \in D, \omega \in \Omega  }
		&
		\overline{ \delta(q_{1},\sigma,\theta_{1},q',+1,d,\omega) } 
		\delta(q_{2},\sigma,\theta_{2},q',0,d,\omega) \\
		& + \overline{ \delta(q_{1},\sigma,\theta_{1},q',0,d,\omega) } 
		\delta(q_{2},\sigma,\theta_{2},q',-1,d,\omega) = 0.
	\end{array} 
\end{equation}
3. $ x_{1} = x_{2} $ and $ k_{1} = k_{2}-2 $ ($ x_{1} = x_{2} $ and $ k_{1} = k_{2}+2 $):
\begin{equation}
	\sum\limits_{q' \in Q, d \in D, \omega \in \Omega  }
	\overline{ \delta(q_{1},\sigma,\theta_{1},q',+1,d,\omega) } 
		\delta(q_{2},\sigma,\theta_{2},q',-1,d,\omega) = 0.
\end{equation}
4. $ x_{1} = x_{2}-1 $  and $ k_{1} = k_{2} $ ($ x_{1} = x_{2}+1 $ and $ k_{1} = k_{2} $):
\begin{equation}\label{equation:appeq4}
	\sum\limits_{q' \in Q, c \in \lozenge, \omega \in \Omega  }
	\overline{ \delta(q_{1},\sigma_{1},\theta,q',c,\rightarrow,\omega) } 
		\delta(q_{2},\sigma_{2},\theta,q',c,\downarrow,\omega) = 0.
\end{equation}
5. $ x_{1} = x_{2}-1 $ and $ k_{1} = k_{2}-1 $ ($ x_{1} = x_{2}+1 $ and $ k_{1} = k_{2}+1 $):
\begin{equation}
	\begin{array}{ll}
		\sum\limits_{q' \in Q, \omega \in \Omega  }
		&
		\overline{ \delta(q_{1},\sigma_{1},\theta_{1},q',+1,\rightarrow,\omega) } 
		\delta(q_{2},\sigma_{2},\theta_{2},q',0,\downarrow,\omega) \\
		& + \overline{ \delta(q_{1},\sigma_{1},\theta_{1},q',0,\rightarrow,\omega) } 
		\delta(q_{2},\sigma_{2},\theta_{2},q',-1,\downarrow,\omega) = 0.
	\end{array} 
\end{equation}
6. $ x_{1} = x_{2}-1 $ and $ k_{1} = k_{2}+1 $ ($ x_{1} = x_{2}+1 $ and $ k_{1} = k_{2}-1 $):
\begin{equation}
	\begin{array}{ll}
		\sum\limits_{q' \in Q, \omega \in \Omega  }
		&
		\overline{ \delta(q_{1},\sigma_{1},\theta_{1},q',-1,\rightarrow,\omega) } 
		\delta(q_{2},\sigma_{2},\theta_{2},q',0,\downarrow,\omega) \\
		& + \overline{ \delta(q_{1},\sigma_{1},\theta_{1},q',0,\rightarrow,\omega) } 
		\delta(q_{2},\sigma_{2},\theta_{2},q',+1,\downarrow,\omega) = 0.
	\end{array} 
\end{equation}
7. $ x_{1} = x_{2}-1 $ and $ k_{1} = k_{2}-2 $ ($ x_{1} = x_{2}+1 $ and $ k_{1} = k_{2}+2 $):
\begin{equation}
	\sum\limits_{q' \in Q, \omega \in \Omega  }
	\overline{ \delta(q_{1},\sigma_{1},\theta_{1},q',+1,\rightarrow,\omega) } 
		\delta(q_{2},\sigma_{2},\theta_{2},q',-1,\downarrow,\omega) = 0.
\end{equation}
8. $ x_{1} = x_{2}-1 $ and $ k_{1} = k_{2}+2 $ ($ x_{1} = x_{2}+1 $ and $ k_{1} = k_{2}-2 $):
\begin{equation}\label{equation:lastapp}
	\sum\limits_{q' \in Q, \omega \in \Omega  }
	\overline{ \delta(q_{1},\sigma_{1},\theta_{1},q',-1,\rightarrow,\omega) } 
		\delta(q_{2},\sigma_{2},\theta_{2},q',+1,\downarrow,\omega) = 0.
\end{equation}
\bibliographystyle{plain}
\bibliography{SAYYAKARYILMAZ}

\begin{thebibliography}{10}

\bibitem{AI99}
Masami Amano and Kazuo Iwama.
\newblock Undecidability on quantum finite automata.
\newblock In {\em STOC'99: Proceedings of the thirty-first annual ACM symposium
  on Theory of computing}, pages 368--375. ACM, 1999.

\bibitem{BFK01}
Richard Bonner, R\={u}si\c{n}\v{s} Freivalds, and Maksim Kravtsev.
\newblock Quantum versus probabilistic one-way finite automata with counter.
\newblock In {\em SOFSEM 2001: Theory and Practice of Computer Science}, volume
  2234 of {\em Lecture Notes in Computer Science}, pages 181--190, 2001.

\bibitem{Bo03}
Symeon Bozapalidis.
\newblock Extending stochastic and quantum functions.
\newblock {\em Theory of Computing Systems}, 36(2):183--197, 2003.

\bibitem{DS90}
Cynthia Dwork and Larry Stockmeyer.
\newblock A time complexity gap for two-way probabilistic finite-state
  automata.
\newblock {\em SIAM Journal on Computing}, 19(6):1011--1123, 1990.

\bibitem{Hi10}
Mika Hirvensalo.
\newblock Quantum automata with open time evolution.
\newblock {\em International Journal of Natural Computing Research},
  1(1):70--85, 2010.

\bibitem{HS10}
Juraj Hromkovi\v{c} and Georg Schnitger.
\newblock On probabilistic pushdown automata.
\newblock {\em Information and Computation}, 208:982--995, 2010.

\bibitem{Je07}
Emmanuel Jeandel.
\newblock Topological automata.
\newblock {\em Theory of Computing Systems}, 40(4):397--407, 2007.

\bibitem{Ka91}
J\={a}nis Ka{\c{n}}eps.
\newblock Stochasticity of the languages acceptable by two-way finite
  probabilistic automata.
\newblock {\em Discrete Mathematics and Applications}, 1:405--421, 1991.

\bibitem{KW97}
Attila Kondacs and John Watrous.
\newblock On the power of quantum finite state automata.
\newblock In {\em FOCS'97: Proceedings of the 38th Annual Symposium on
  Foundations of Computer Science}, pages 66--75, 1997.

\bibitem{Kr99}
Maksim Kravtsev.
\newblock Quantum finite one-counter automata.
\newblock In {\em SOFSEM'99: Theory and Practice of Computer Science}, volume
  1725 of {\em Lecture Notes in Computer Science}, pages 431--440, 1999.

\bibitem{NHK06}
Masaki Nakanishi, Kiyoharu Hamaguchi, and Toshinobu Kashiwabara.
\newblock Expressive power of quantum pushdown automata with classical stack
  operations under the perfect-soundness condition.
\newblock {\em IEICE Transactions}, E89-D(3):1120--1127, March 2006.

\bibitem{NH71}
Masakazu Nasu and Namio Honda.
\newblock A context-free language which is not acceptable by a probabilistic
  automaton.
\newblock {\em Information and Control}, 18(3):233--236, 1971.

\bibitem{NC00}
Michael~A. Nielsen and Isaac~L. Chuang.
\newblock {\em Quantum Computation and Quantum Information}.
\newblock Cambridge University Press, 2000.

\bibitem{Pe94}
H.~Petersen.
\newblock Two-way one-counter automata accepting bounded languages.
\newblock {\em ACM SIGACT News}, 25(3):102--105, September 1994.
\newblock DOI: 10.1145/193820.193835.

\bibitem{Ra63B}
Michael~O. Rabin.
\newblock Real-time computation.
\newblock {\em Israel Journal of Mathematics}, 1(4), 1963.

\bibitem{Ra92}
Bala Ravikumar.
\newblock Some observations on 2-way probabilistic finite automata.
\newblock In {\em Proceedings of the 12th Conference on Foundations of Software
  Technology and Theoretical Computer Science}, pages 392--403.
  Springer-Verlag, 1992.

\bibitem{SYY10}
A.~C.~Cem Say, Abuzer Yakary{\i}lmaz, and {\c S}efika Y\"{u}zsever.
\newblock Quantum one-way one-counter automata.
\newblock In R\={u}si\c{n}\v{s} Freivalds, editor, {\em Randomized and quantum
  computation}, pages 25--34, 2010.
\newblock Satellite workshop of MFCS and CSL 2010.

\bibitem{SHL65}
Richard~Edwin Stearns, Juris Hartmanis, and Philip M.~Lewis II.
\newblock Hierarchies of memory limited computations.
\newblock In {\em IEEE Conference Record on Switching Circuit Theory and
  Logical Design}, pages 179--190, 1965.

\bibitem{Wa03}
John Watrous.
\newblock On the complexity of simulating space-bounded quantum computations.
\newblock {\em Computational Complexity}, 12(1-2):48--84, 2003.

\bibitem{YS10A}
Abuzer Yakary{\i}lmaz and A.~C.~Cem Say.
\newblock Languages recognized by nondeterministic quantum finite automata.
\newblock {\em Quantum Information and Computation}, 10(9\&10):747--770, 2010.

\bibitem{YS10B}
Abuzer Yakary{\i}lmaz and A.~C.~Cem Say.
\newblock Succinctness of two-way probabilistic and quantum finite automata.
\newblock {\em Discrete Mathematics and Theoretical Computer Science},
  12(4):19--40, 2010.

\bibitem{YS11A}
Abuzer Yakary{\i}lmaz and A.~C.~Cem Say.
\newblock Unbounded-error quantum computation with small space bounds.
\newblock {\em Information and Computation}, 209(6):873--892, 2011.

\bibitem{YFSA11A}
Abuzer Yakaryılmaz, R\={u}si\c{n}\v{s} Freivalds, A.~C.~Cem Say, and Ruben
  Agadzanyan.
\newblock Quantum computation with write-only memory.
\newblock {\em Natural Computing}, 2011.
\newblock DOI: 10.1007/s11047-011-9270-0.

\bibitem{YKI05}
Tomohiro Yamasaki, Hirotada Kobayashi, and Hiroshi Imai.
\newblock Quantum versus deterministic counter automata.
\newblock {\em Theoretical Computer Science}, 334(1-3):275--297, 2005.

\bibitem{YKTI02}
Tomohiro Yamasaki, Hirotada Kobayashi, Yuuki Tokunaga, and Hiroshi Imai.
\newblock One-way probabilistic reversible and quantum one-counter automata.
\newblock {\em Theoretical Computer Science}, 289(2):963--976, 2002.

\end{thebibliography}

\end{document}